\documentclass[sigconf]{aamas}

\usepackage{balance} %

\usepackage{blindtext}

\usepackage{graphicx}
\usepackage{adjustbox}
\usepackage{wrapfig}
\usepackage[labelfont=bf,figurename=Fig.,font=small]{caption}
\usepackage{subcaption}
\usepackage{lipsum}
\usepackage{amsthm}
\usepackage{amsmath}
\usepackage{siunitx}
\usepackage{mathtools}
\usepackage{xcolor}
\usepackage{hyperref}
\usepackage[ruled,vlined,linesnumbered]{algorithm2e}
\usepackage[noend]{algpseudocode}
\usepackage{ifthen}
\usepackage{dsfont}
\usepackage{dirtytalk}
\usepackage{afterpage}

\allowdisplaybreaks

\newcommand{\state}{x}
\newcommand{\statedyn}{f}

\newcommand{\outputvals}{y}
\newcommand{\ctrl}{u}
\newcommand{\ctrldim}{m}

\newcommand{\costcomponent}{g}
\newcommand{\timehorizon}{T}
\newcommand{\rewardparam}{\theta}

\newcommand{\ra}{\rightarrow}
\newcommand{\R}{\mathbb{R}}

\newcommand{\feedback}{\gamma}
\newcommand{\feedbackset}{\Gamma}

\newcommand{\cost}{J}

\newtheorem{remark}{Remark}

\newtheorem{proposition}{Proposition}
\newtheorem{definition}{Definition}

\let\originalleft\left
\let\originalright\right
\renewcommand{\left}{\mathopen{}\mathclose\bgroup\originalleft}
\renewcommand{\right}{\aftergroup\egroup\originalright}

\title{\LARGE \bf Cost Inference for Feedback Dynamic Games from Partial and Incomplete Observations
}

\setcopyright{ifaamas}
\acmConference[AAMAS '23]{Proc.\@ of the 22nd International Conference
on Autonomous Agents and Multiagent Systems (AAMAS 2023)}{May 29 -- June 2, 2023} 
{London, United Kingdom}{A.~Ricci, W.~Yeoh, N.~Agmon, B.~An (eds.)}
\copyrightyear{2023}
\acmYear{2023}
\acmDOI{}
\acmPrice{}
\acmISBN{}

\acmSubmissionID{1052}

\title[Cost Inference for Feedback Dynamic Games]{Cost Inference for Feedback Dynamic Games from\\Noisy Partial State Observations and Incomplete Trajectories}

\author{Jingqi Li}
\affiliation{
  \institution{University of California, Berkeley}
  \city{Berkeley}
  \country{United States}}
\email{jingqili@berkeley.edu}

\author{Chih-Yuan Chiu}
\affiliation{
  \institution{University of California, Berkeley}
  \city{Berkeley}
  \country{United States}}
\email{chihyuan\_chiu@berkeley.edu}

\author{Lasse Peters}
\affiliation{
  \institution{Delft University of Technology}
  \city{Delft}
  \country{Netherlands}}
\email{l.peters@tudelft.nl}

\author{Somayeh Sojoudi}
\affiliation{
  \institution{University of California, Berkeley}
  \city{Berkeley}
  \country{United States}}
\email{sojoudi@berkeley.edu}

\author{Claire Tomlin}
\affiliation{
  \institution{University of California, Berkeley}
  \city{Berkeley}
  \country{United States}}
\email{tomlin@eecs.berkeley.edu}

\author{David Fridovich-Keil}
\affiliation{
  \institution{University of Texas, Austin}
  \city{Austin}
  \country{United States}}
\email{dfk@utexas.edu}

\begin{abstract}
In multi-agent dynamic games, the Nash equilibrium state trajectory of each agent is determined by its cost function and the information pattern of the game. However, the cost and trajectory of each agent may be unavailable to the other agents. Prior work on using partial observations to infer the costs in dynamic games assumes an open-loop information pattern. In this work, we demonstrate that the feedback Nash equilibrium concept is more expressive and encodes more complex behavior. It is desirable to develop specific tools for inferring players' objectives in feedback games. Therefore, we consider the dynamic game cost inference problem under the feedback information pattern, using only partial state observations and incomplete trajectory data. To this end, we first propose an inverse feedback game loss function, whose minimizer yields a feedback Nash equilibrium state trajectory closest to the observation data. We characterize the landscape and differentiability of the loss function. Given the difficulty of obtaining the exact gradient, our main contribution is an efficient gradient approximator, which enables a novel inverse feedback game solver that minimizes the loss using first-order optimization. In thorough empirical evaluations, we demonstrate that our algorithm converges reliably and has better robustness and generalization performance than the open-loop baseline method when the observation data reflects a group of players acting in a feedback Nash game. %

\end{abstract}

\keywords{Inverse Games, Dynamic Game Theory, Nash Equilibrium}

\newcommand{\BibTeX}{\rm B\kern-.05em{\sc i\kern-.025em b}\kern-.08em\TeX}

\begin{document}

\pagestyle{fancy}
\fancyhead{}

\maketitle

\section{Introduction}
\label{sec: Introduction}

The safety and efficiency of urban traffic relies heavily on the ability of each participant to predict the effects of their actions on others' decisions \cite{molloy2018inverse,schwarting2021stochastic}.
For example, drivers on a highway may wish to halt an overtaking maneuver if they believe the other drivers are aggressive, and some drivers may decelerate their cars to avoid collision if they believe that another driver wishes to merge. 

A powerful paradigm for modeling the  interdependence of decisions in multi-agent settings is provided \textit{general-sum dynamic games} \cite{basar1998DynamicNoncooperativeGameTheory,isaacs1999differential}.
A Nash equilibrium solution of a game-theoretic model can be used to \emph{simultaneously} predict the actions of all agents in the scene.
This equilibrium solution is particularly expressive when the game possesses a feedback information structure.
In this case, each equilibrium strategy explicitly accounts for the dynamically evolving information available to each player over time. 

Despite the theoretical attractiveness of this modeling paradigm, in reality, autonomous agents often have only limited information available about the world around them.
For example, in urban traffic an autonomous agent typically has incomplete knowledge of the objectives of other players. To address this challenge, recent works on \textit{inverse dynamic game theory} \cite{rothfuss2017inverse,peters2021inferring,molloy2022inversebook} recover these objectives from past trajectory data.
Moreover, in realistic applications, only noisy sensor measurements of agents' states are available.
This partial observability further complicates the inverse game problem, and existing work \cite{peters2021inferring} treats this case in the open-loop information structure.

In this work, we present a gradient-based solver for inverse dynamic games, under the state feedback information structure. Our solver can recover objectives from partial state observations of incomplete trajectories. Both of these effects are common in robotics due to noisy perception and occlusions. 
We show that our algorithm converges reliably in practice, and demonstrate the superior robustness and generalization performance as compared with a baseline method which learns cost functions under the open-loop assumption \cite{peters2021inferring}, when the observation data is from a group of players pursuing a feedback Nash equilibrium strategy. 

Our contributions are threefold. Firstly, we characterize the solution set of the inverse feedback dynamic game problem. In particular, we show that the set of the global minima could be nonconvex and disconnected, and discuss regularization schemes to mitigate this problem. Secondly, we show the differentiability of the loss function in linear quadratic games and propose a computationally efficient procedure to approximate the gradient for nonlinear games. Finally, we propose an efficient first-order  coordinate-descent solver for the inverse feedback game problem, using noisy partial observations of an incomplete expert state trajectory. Experimental results show that our method reliably converges for inverse feedback games with nonlinear dynamics and is able to learn nonconvex costs. Moreover, the converged cost function can accurately predict the feedback Nash equilibrium state trajectories even for unseen initial states.

\section{RELATED WORK}
\label{sec: Related Work}

\subsection{Non-cooperative Dynamic Games}
\label{subsec: Multi-Player Path Planning via Dynamic Games}
Non-cooperative dynamic game theory \cite{basar1998DynamicNoncooperativeGameTheory,isaacs1999differential} provides a formal framework for analyzing strategic interaction in a multi-agent setting \cite{cruz1975survey,basar1998DynamicNoncooperativeGameTheory,lee2008human}. In non-cooperative games, each player minimizes its own individual cost function; since players' costs may not be mutually aligned, the resulting equilibrium behavior is generally competitive. Among different equilibrium concepts, the Nash equilibrium has been extensively studied because of its representative power of capturing many non-cooperative behaviors arising in real-world multi-agent systems~\cite{gabler2017game,schwarting2019social}.

Recent advances in the literature aim to develop efficient solutions to Nash equilibrium problems in dynamic games. Though the solutions for the open-loop and feedback Nash equilibrium in linear quadratic (LQ) games are well understood \cite{basar1998DynamicNoncooperativeGameTheory}, for nonlinear games there is no closed-form solution in general. The work \cite{ratliff2016characterization} characterizes the local Nash solution concept for open-loop Nash equilibrium. In the feedback setting, numerous approaches have been proposed under various special cases \cite{tanwani2019feedback,kossioris2008feedback}. A value iteration based approach for computing feedback Nash equilibria of nonlinear games without constraints is introduced in \cite{herrera2019algorithm}. Recently, a set of KKT conditions for feedback Nash equilibria in constrained nonlinear games is derived in \cite{laine2021computation}. Computing a feedback Nash equilibrium is challenging due to the nested KKT conditions in different time steps.

Our work draws upon the ILQGames \cite{fridovich2020efficient} framework, which at each iteration solves a linear-quadratic game that approximates the original game. The construction of the approximate game parallels the iterative linearization and quadraticization methods of iterative LQR \cite{LiTodorov2004iterative}, and the dynamic programming equations that characterize equilibrium strategies in linear quadratic dynamic games \cite{basar1998DynamicNoncooperativeGameTheory}. This approach differs from the ALGames \cite{cleac2019algames} method, which computes an open-loop Nash equilibrium strategy.

\subsection{\emph{Inverse} Non-cooperative Dynamic Games}
\label{subsec: Cost, Belief, and Information Pattern Inference}

In contrast to the forward game problem of computing a strategy in dynamic games, an inverse game problem amounts to finding objectives for all agents such that the corresponding strategic (e.g., Nash equilibrium) interactions reproduce expert demonstrations. The inverse game problem is important because it paves the way for an agent to understand the preferences which explain other agents' behavior, which may facilitate more efficient multi-agent interaction and coordination. %

The problem of inverse infinite-horizon LQ games is considered in \cite{inga2019solution}, where the set of cost functions whose feedback Nash equilibrium strategies coincide with an expert strategy is derived. In \cite{rothfuss2017inverse,yu2022inverse}, the two-player inverse LQ game is solved by transforming the problem to an inverse optimal control under the assumption that the control input data of one player is known. 
Two methods based on the KKT conditions of an open-loop Nash equilibrium are proposed for open-loop general-sum differential games in \cite{molloy2019inverse}. Several necessary conditions for open-loop Nash equilibria are proposed in \cite{MOLLOY201711788} and used for developing an inverse game solution for some classes of open-loop games.

Recently, an efficient bilevel optimization framework \cite{peters2021inferring} based on the open-loop Nash equilibrium KKT conditions was proposed for solving inverse games with an open-loop Nash assumption. Another line of work on inferring costs in open-loop games \cite{awasthi2020inverse,inga2019inverse,englert2017inverse} proposes to minimize the residual violation of the KKT conditions. This KKT residual framework assumes the knowledge of complete trajectory data and is a convex problem. Given the difficulty of evaluating KKT conditions for feedback Nash equilibria in nonlinear games \cite{laine2021computation}, the extension of the KKT residual method to feedback nonlinear games may be subject to numerical difficulty. 

A bilevel optimization approach for inverse feedback game problem is proposed in \cite{molloy2022inverse}, with the assumption that both the expert state and control trajectories are observed without noise. In addition, an inverse game solver is proposed in \cite{Mehr2021MaximumEntropyMultiAgentDynamicGames} where they infer the players' cost functions with the assumption that the expert strategy follows a new concept called Maximum Entropy Nash Equilibrium. 
To the best of the authors' knowledge, there is no work on inferring cost functions of nonlinear dynamic games under feedback Nash equilibrium condition, from noisy partial state observation and incomplete trajectory data.

\section{PRELIMINARIES}
\label{sec: Preliminaries}

Consider an $N$-player, $T$-stage, deterministic, discrete-time dynamic game, with a state $x_t^i\in\mathbb{R}^{n_i}$ and control input $\ctrl^i_t \in \R^{\ctrldim_i}$ for each player $i \in [N]:= \{1, \cdots, N\}$, $t\in[T]$. Let the dimension of the joint state and control input be $n:=\sum_{i=1}^N n_i$ and $m:=\sum_{i=1}^N m_i$, respectively. We denote by $\state_t:=[\state_t^1,\dots,\state_t^N]\in\mathbb{R}^n$ and $\ctrl_t:=[\ctrl_t^1,\dots,\ctrl_t^N]\in\mathbb{R}^m$ the joint state and joint control at time $t\in[T]$, respectively. The joint dynamics for the system is given by the differentiable dynamics map $f_t(\cdot,\cdot): \R^n \times \R^{m}\to \mathbb{R}^n$:
\begin{align}\label{eq:general_game_dynamics}
    \state_{t+1} = \statedyn_t(\state_t, \ctrl_t), \hspace{6mm} \forall \hspace{0.5mm} t = 1, \cdots, T.
\end{align}
We denote by $\mathbf{f}:=\{f_t\}_{t=1}^T$ the set of dynamics across all the time instances within horizon $T$. We define $\mathbf{\state}:=\{\state_t\}_{t=1}^{T}$ and $\mathbf{\ctrl}:=\{\ctrl_t\}_{t=1}^T$ to be a state trajectory and control trajectory, respectively, if $x_{t+1}=f(x_t,u_t)$, for each $t\in[T]$. The objective of each agent $i$ is to minimize its overall cost, given by the sum of its running costs $\costcomponent_t^i: \R^n \times \R^{m} \ra \R$ over the time horizon:
\begin{align} \label{Eqn: Cost, integrated}
    \cost^i(\mathbf{\state},\mathbf{\ctrl}) := \sum_{t=1}^{T} \costcomponent^i_t(\state_t, \ctrl_t)
\end{align}
Define $\costcomponent_{t}:=\{\costcomponent_{t}^1,\costcomponent_t^2,\cdots, \costcomponent_{t}^N\}$, $t\in[T]$. We denote by $\mathbf{g}:=\{g_t\}_{t=1}^T$ the set of cost functions for all the agents within horizon $T$.

To minimize \eqref{Eqn: Cost, integrated}, each player uses their observations of the environment to design a sequence of control inputs to deploy during the discrete time interval $[T]$. The information available to each player at each time characterizes the \textit{information pattern} of the dynamic game, and plays a major role in shaping the optimal responses of each player \cite{basar1998DynamicNoncooperativeGameTheory}. Below, we explore two such information patterns---\textit{feedback} and \textit{open-loop}.

\subsection{Nash Solutions in Feedback Strategies}
\label{subsec: Feedback Strategies and Nash Equilibrium}

Under the state feedback information pattern, each player observes the state $\state_t$ at each time $t$, and uses this information to design a \textit{feedback strategy} $\feedback_t^i: \R^n \ra \R^{m_i}$, given by: $\ctrl_t^i := \gamma_t^i(\state_t)$, for each $i \in [N]$ and $t \in [T]$. Let $\feedback_t(x_t):=[\feedback_t^1(x_t),\feedback_t^2(x_t),\dots,\feedback_t^N(x_t)]\in\mathbb{R}^m$.

Following the notation of \cite{basar1998DynamicNoncooperativeGameTheory}, we denote by $\Gamma_t^i$ the set of all state feedback strategies of player $i$, for each $i \in [N]$. Under this \textit{feedback} information pattern, the Nash equilibrium of the dynamic game is as defined below.

\begin{definition}[\textbf{Feedback Nash Equilibrium (FBNE)} {\cite[Ch.~6]{basar1998DynamicNoncooperativeGameTheory}}]
The set of control strategies $\{{\gamma_t^{1*}}, \cdots, {\gamma_t^{N*}}\}_{t=1}^T$ is called a \textit{feedback Nash equilibrium} if no player is incentivized to unilaterally alter its strategy. Formally:
\begin{align}\color{red} \label{Eqn: Feedback Nash Equilibrium}
    &{W_t^{i*}}\left({\state_t,[\feedback_{t}^{1*}}(\state_t), \ldots,  {\feedback_t^{i*}}(\state_t), \ldots, {\feedback_t^{N*}}(\state_t)]\right) \\ \nonumber
    & \leq {W_t^{i*}}\left(\state_t,[{\feedback_t^{1*}}(\state_t), \ldots, {\feedback_t^i}(\state_t),  \ldots, {\feedback_t^{N*}}(\state_t)]\right), \forall \feedback_t^{i} \in \feedbackset_t^{i}, \forall t\in[T].
\end{align}
where ${W_t^{i*}}(\cdot,\cdot):\mathbb{R}^n\times \mathbb{R}^m\to \mathbb{R}$, $t\in[T]$ is the optimal state-action function defined as follows, 
\begin{equation}
    \begin{aligned}
    {W_T^{i*}}(x_T,u_T)&:= g_T^i(x_T,u_T)\\
    {W_t^{i*}}(x_t,u_t)& := g_t^i(x_t,u_t) + {V^{i*}_{t+1}}(x_{t+1}),\forall t\in[T-1],\\
    {V_t^{i*}}(x_t)&:= {W_t^{i*}}(x_t,[{\gamma_t^1}^*(x_t),\dots, {\gamma_t^N}^*(x_t)]), \forall t\in[T].
    \end{aligned}
\end{equation}
\end{definition}
We define $\mathbf{\state}$ and $\mathbf{\ctrl}$ to be a FBNE state trajectory and a FBNE control trajectory, respectively, if $u_t^i={\gamma_t^{i*}}(x_t)$, for each $i\in[N]$ and $t\in [T]$. We denote by $\xi(\mathbf{f}, \mathbf{\costcomponent})$ the set of all FBNE state trajectories in the game defined by the dynamics $\mathbf{f}$ and cost functions $\mathbf{g}$. 

\begin{remark}[Strong Time Consistency]\label{remark:strong time-consistency}
The FBNE conditions of \eqref{Eqn: Feedback Nash Equilibrium} implicitly enforce strong time-consistency \cite[Def. 5.14]{basar1998DynamicNoncooperativeGameTheory} of the equilibrium strategies.
That is, FBNE does not admit arbitrary feedback strategies, but imposes the additional condition that those strategies must also be in equilibrium for any subgame starting at a later stage from an arbitrary state.
\end{remark}

\subsection{Nash Solutions in Open-loop Strategies}
\label{subsec: Open-Loop Strategies and Nash Equilibrium}

In contrast, under the open-loop information pattern, each player only observes the initial state $\state_1$. In this case, the strategy for each player $i \in [N]$ is a map from $x_1$ to $\{u_1^i,u_2^i,\cdots, u_T^i\}$, which we denote by $\phi^i(\cdot):\mathbb{R}^n\to \underbrace{\mathbb{R}^{m_i}\times \cdots \times \mathbb{R}^{m_i}}_T$. Let $\Phi^i$ be the set of all open-loop strategies of the player $i$, $i\in [N]$. The corresponding \textit{open-loop Nash equilibrium} is defined as follows.
 
\begin{definition}[\textbf{Open-Loop Nash Equilibrium (OLNE)} {\cite[Ch.~6]{basar1998DynamicNoncooperativeGameTheory}}] \label{Def: Open-Loop Nash Equilibrium}
The tuple of control strategies $\{\phi_1^*, \cdots, \phi_N^*\}$ is called an \textit{open-loop Nash equilibrium} if no player is incentivized to unilaterally alter its sequence of control inputs. Formally:
\begin{align} \label{Eqn: Open-Loop Nash Equilibrium}
    &\cost^{i}\left(\mathbf{\state}, [{\phi^1}^*(x_1), \cdots, {\phi^{i}}^*(x_1), \cdots, {\phi^{N}}^*(x_1)]\right) \\ \nonumber
    \leq \hspace{0.5mm} &\cost^{i}\left(\mathbf{\state},[{\phi^1}^*(x_1), \cdots, {\phi^{i}}(x_1), \cdots,  {\phi^{N}}^*(x_1)]\right),
    \forall \phi^{i} \in \Phi^i, \forall x_1\in\mathbb{R}^n.
\end{align}
\end{definition}

\begin{remark}\label{remark:weak time-consistency}
The OLNE definition does not imply the strong time-consistence as in the feedback counterpart \cite{basar1998DynamicNoncooperativeGameTheory}.
\end{remark}

\subsection{Feedback vs. Open-loop Nash Equilibria}
In this subsection, we demonstrate the difference between open-loop and feedback Nash equilibria and show the necessity of developing specific solutions for cost inference problems with the feedback information pattern, instead of applying existing work with the open-loop assumption \cite{peters2020inference}. To this end, we introduce below several linear-quadratic (LQ) games where the open-loop Nash equilibrium (OLNE) and feedback Nash equilibrium (FBNE) state trajectories differ substantially. 
\begin{figure}[t!]
    \centering
    \includegraphics[width=0.48\textwidth, trim = 0cm 5.5cm 0cm 4.5cm]{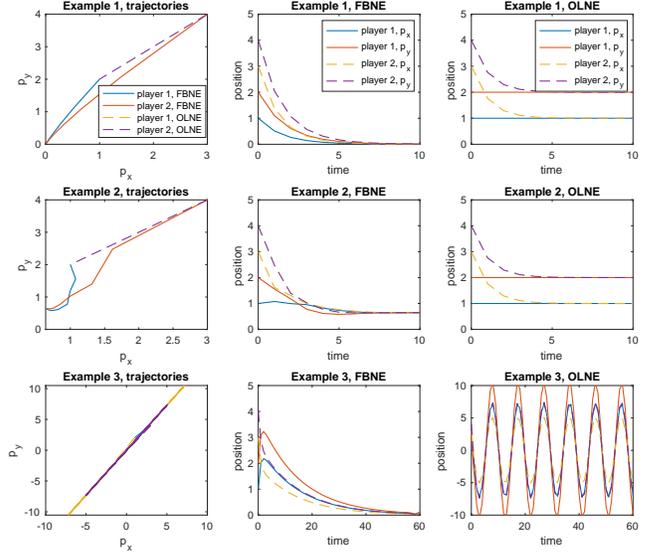}
    \caption{Examples of cost functions that yield trajectories that are different under the OLNE and FBNE assumptions.} %
    \label{fig:proposition 1}
\end{figure}
LQ games are a class of dynamic games with dynamics and player objectives of the form in \eqref{eq:LQ_dynamics} and \eqref{eq:LQ_costs}, respectively,
\begin{equation}\label{eq:LQ_dynamics}
    \begin{aligned}
    \state_{t+1}=A_t \state_t + \sum_{i\in[N]} B_t^i \ctrl_t^i,\ \forall t\in[T],
    \end{aligned}
\end{equation}
\begin{equation}\label{eq:LQ_costs}
    \costcomponent_t^i(\state_t, \ctrl_t) =\frac{1}{2} (\state_t^\top Q^i_t \state_t +\sum_{j\in[N]} {\ctrl_t^j}^\top R^{ij}_t \ctrl_t^j), \forall t\in[T],\forall i\in [N],
\end{equation} 
where matrices $\{A_t, B_t^i\}$, positive semidefinite matrix $Q_t^i$ and positive definite matrix $R_t^{ij}$ are defined with appropriate dimensions, for each $i,j\in[N]$ and $t\in[T]$. 

\textbf{Case Study:} We consider a two-player LQ game with a state vector $x_t=[p_{x,t}^1, p_{y,t}^1, p_{x,t}^2,p_{y,t}^2]$, where $p_{x,t}^i$ and  $p_{y,t}^i$ are the x- and y-coordinates of agent $i \in \{1,2\}$, respectively. Let $\ctrl_t^i \in \R^2$ be the control input for the $i$-th agent, $i\in\{1,2\}$. In this setting, we consider a class of games in which the first agent wants to drive the second agent to the origin, while the second agent wants to catch the first agent. The agents' joint dynamics and costs at time $t\in[T]$ are specified as follows:
\begin{equation}
    \begin{aligned}      
    \state_{t+1} &= \begin{bmatrix} I_2 & 0\\ 0 & I_2 \end{bmatrix} \state_t + \begin{bmatrix} I_2 \\ 0 \end{bmatrix} \ctrl_t^1 + \begin{bmatrix} 0\\ I_2 \end{bmatrix} \ctrl_t^2,\\
      \costcomponent_t^1(\state_t,\ctrl_t) &= \|p_{x,t}^2\|_2^2 + \|p_{y,t}^2\|_2^2 + \|\ctrl_t^1\|_2^2, \\  
      \costcomponent_t^2(\state_t, \ctrl_t) &= \|p_{x,t}^2 - p_{x,t}^1\|_2^2 + \|p_{y,t}^2 - p_{y,t}^1\|_2^2 + \|\ctrl_t^2\|_2^2,
    \end{aligned}\label{eq:LQ_counter_example}
\end{equation}
where $I_2$ is the $2 \times 2$ identity matrix. We visualize the unique FBNE and OLNE state trajectories of this example in the first row in Fig. \ref{fig:proposition 1}. If we modify the cost function of the first player such that it wants to lead the $x$- and $y$-position of the second player to be aligned with each other, i.e., 
\begin{equation}
    \hat{\costcomponent}_t^1(\state_t,\ctrl_t) := \|p_{x,t}^2-p_{y,t}^2\|_2^2 + \|u_t^1\|_2^2,
\end{equation}
then, the unique FBNE and OLNE state trajectories are still different, as shown in the second row of Fig. \ref{fig:proposition 1}. Moreover, observations of players may be noisy in practice. To illustrate this, we consider a task where the two agents want to catch each other, but the first player's observation of the second player's position is inaccurate. We modify the first player's cost in \eqref{eq:LQ_counter_example} as follows:
\begin{equation}
    \hat{\hat{\costcomponent}}_t^1(\state_t,\ctrl_t):=\|p_{x,t}^1-2p_{x,t}^2\|_2^2 + \|p_{y,t}^1-2p_{y,t}^2\|_2^2 + \|u_t^1\|_2^2 .
\end{equation}
The third row of Fig. \ref{fig:proposition 1} reveals that the FBNE state trajectory is robust to inaccurate observations, but the unique OLNE state trajectory is not.

Thus, it is readily apparent that the OLNE and FBNE state strategies can be substantially different even for fixed cost functions. This difference in expressive power may be understood as a consequence of the strong time consistency property, which is enforced in the feedback information structure but not in the open-loop setting, per Remarks \ref{remark:strong time-consistency} and \ref{remark:weak time-consistency}.
A similar problem arises in the cost inference problem, where the existing OLNE cost inference algorithms may fail to infer the correct cost function in feedback games.

\section{Problem Statement}
\label{sec: Problem Statement}

Let $\mathbf{x}$ be an expert FBNE state trajectory under the nonlinear dynamics $\mathbf{f}$ but unknown cost functions $\{g_t^i\}_{t=1,i=1}^{T,N}$. Let $\mathcal{T}\subseteq [T]$ be the set of observed time indices of the trajectory $\mathbf{x}$. We denote by $\mathbf{y}_{\mathcal{T}}:=\{\outputvals_t\}_{t\in\mathcal{T}}$ the observation data of $\mathbf{x}$, where $y_t\in\mathbb{R}^\ell$ is a partial observation of the state, composed of certain coordinates of $x_t$ corrupted by noise. The task is to infer the cost function of each player such that those inferred costs jointly yield a FBNE state trajectory that is as close as possible to the observed trajectory. We parameterize the cost of the player $i$ by a vector $\theta^i\in\mathbb{R}^{d_i}$, and let $\theta:=[\theta^1,\theta^2,\dots,\theta^N]\in\mathbb{R}^d$. Denote by $g_{t,\theta}^i(x_t,u_t)=\sum_{j=1}^{d_i} \theta_j^i b_{t,j}^i(x_t,u_t)$ player $i$'s parameterized cost at time $t\in[T]$, for some basis functions $\{\{b_{t,j}^i\}_{j=1}^{d_i}\}_{t=1,i=1}^{T,N}$. 
Define $\mathbf{\costcomponent}_\theta:=\{g_{t,\theta}^i\}_{t=1,i=1}^{T,N}$. Formally, this problem is of the form:
\begin{equation}\label{problem_formulation}
    \begin{aligned}
    &\min_{\theta,x_1,\mathbf{\hat{\state}}} \hspace{5mm}
    &&  -p(\mathbf{\outputvals}_{\mathcal{T}}|\mathbf{\hat{\state}}) \\
    &\textrm{s.t. } && \mathbf{\hat{\state}}\in \xi(\mathbf{f}, \mathbf{g}_\theta,x_1),
    \end{aligned}
\end{equation}
where $p(\cdot|\cdot)$ is the likelihood function corresponding to a known sensor model and $\xi(\mathbf{f},\mathbf{g}_\theta,x_1)$ represents the set of state trajectories from the initial condition $x_1\in\mathbb{R}^n$ following a FBNE strategy, under the cost set $\mathbf{\costcomponent}_{\theta}$. Due to the noisy partial observation, $x_1$ is not assumed to be known and instead needs to be inferred as well in \eqref{problem_formulation}. Note that the above formulation can also be extended to the cases where multiple partially observed incomplete trajectories from different initial conditions are available.

\textbf{Running example:} 
We consider a highway platooning scenario where player 1 wants to guide player 2 to a particular lane of the road. The joint state vector is $x_t=[p_{x,t}^1,p_{y,t}^1, \beta_t^1, v_t^1,p_{x,t}^2,p_{y,t}^2, \beta_t^2, v_t^2]$. The time horizon $T=40$. The dynamics model for the player $i$ is:
\begin{equation}\label{eq:dubins_car}
    \begin{bmatrix}
    p_{x,t+1}^i\\p_{y,t+1}^i\\\beta^i_{t+1} \\ v^i_{t+1}
    \end{bmatrix} = \begin{bmatrix}
    p_{x,t}^i\\p_{y,t}^i\\\beta^i_t\\v^i_t
    \end{bmatrix} + \Delta T\begin{bmatrix}
    v_t^i\cos(\beta^i_t)\\v_t^i\sin(\beta_t^i)\\\omega_t^i\\a_t^i
    \end{bmatrix}
\end{equation}
where $\Delta T$ is a time discretization constant and $u_t^i=[\omega_t^i,a_t^i]\in\mathbb{R}^2$ is the control input for player $i\in[N]$. Let $p_x^*$ be the target lane that player 1 wants to guide player 2 to. We parameterize the cost function of the player $i$ by $\theta^i\in\mathbb{R}^2$, 
\begin{align}
    \costcomponent_{t,\theta}^1(\state_t,\ctrl_t) &= \theta_1^1 \|p_{x,t}^1\|_2^2 + \theta_2^1 \| p_{x,t}^2 - p_x^* \|_2^2 + \|u_t^1\|_2^2 \label{eq:2_cars_costs}\\
    \costcomponent_{t,\theta}^2(\state_t,\ctrl_t) &= \theta_1^2\|p_{x,t}^2 - p_{x,t}^1\|_2^2 + \theta_2^2\|v_{t}^2-1\|_2^2 +  \|u_t^2\|_2^2 ,\forall t\in[T].\nonumber
\end{align}
The ground truth solution is $\theta^* = [0,8,4,4]$. We assume that there is a period of occlusion happening from the time index $t=11$ to $t=19$, and the observed time index set is $\mathcal{T}=\{1,2,\dots,10,20,21,\dots,40\}$. Also, it may be difficult for a human driver to measure other vehicles' velocity accurately, and therefore we assume that partial observation data $\mathbf{y}_\mathcal{T}$ excludes the velocity of both cars in the data set, and is further subject to Gaussian noise of standard deviation $\sigma$. The initial condition $x_1$ is not known and needs to be inferred. We visualize the ground truth solution in the first subplot of Fig. \ref{fig:problem_statement} and the noisy incomplete trajectory data in the second subplot of Fig.~\ref{fig:problem_statement}.

The many challenges of the above problem include: (a) partial observation; (b) noisy and incomplete expert trajectory data; and (c) the difficulty of evaluating and differentiating the objective in \eqref{problem_formulation}, due to the challenge of computing a FBNE strategy in nonlinear games \cite{laine2021computation}. In the following sections, we will characterize the complexity of this inverse feedback game problem and propose an efficient solution.

\begin{figure}[t!]
    \centering
    \includegraphics[width=0.4\textwidth, trim=0cm 1.5cm 0cm 0cm]{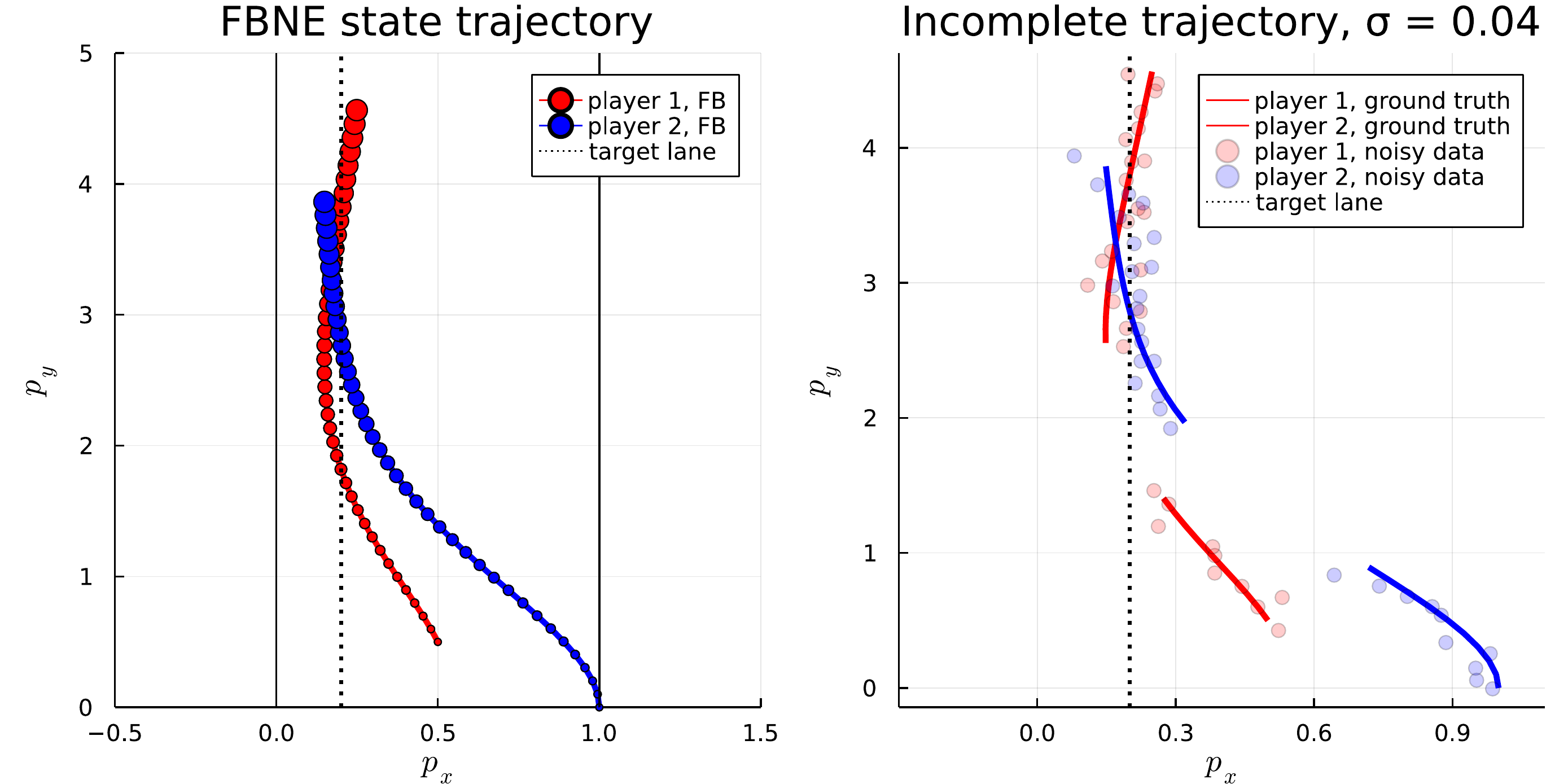}
    \caption{Visualization of the running example.}
    \label{fig:problem_statement}
\end{figure}

\section{Results: From Characterization to Computation}
\label{sec: Methods}
In this section, we first characterize the complexity of the inverse feedback game problem \eqref{problem_formulation}. In particular, we will show the nonconvexity of the loss function and the existence of multiple isolated \emph{global} minima. Based on this observation, we discuss regularization schemes that can mitigate this issue. Our main contribution is to characterize the differentiability of the inverse feedback game loss function in \eqref{problem_formulation}. Finally, we present a gradient approximation scheme that can be used in a first-order optimization formulation.

\subsection{Characterization of the Inverse Feedback Dynamic Game Problem}
The inverse feedback dynamic game problem \eqref{problem_formulation} is a constrained optimization problem, which is hard to solve due to the nonconvexity of the set $\xi(\mathbf{f},\mathbf{g}_\theta,x_1)$. %
With a slight abuse of notation, we denote by $\hat{\mathbf{x}}(\mathbf{f},\mathbf{g}_\theta,x_1) \in \xi(\mathbf{f}, \mathbf{g}_\theta, x_1)$ a FBNE state trajectory. To simplify the problem, we transform \eqref{problem_formulation} to an unconstrained problem by substituting a forward game solution $\hat{\mathbf{x}}(\mathbf{f},\mathbf{g}_\theta,x_1)$ into the likelihood function $p(\mathbf{y}_\mathcal{T}|\hat{\mathbf{\state}})$, as follows:
\begin{equation}\label{eq:from_constrained_problem_to_unconstrained_one}
    \hat{L}(\theta,x_1) :=- p(\mathbf{y}_{\mathcal{T}}|\hat{\mathbf{\state}}(\mathbf{f},\mathbf{g}_\theta,x_1)).
\end{equation}
The minimizer of \eqref{eq:from_constrained_problem_to_unconstrained_one} is a local optimum to the original problem \eqref{problem_formulation} and becomes global when $\xi(\mathbf{f},\mathbf{\costcomponent}_\theta,x_1)$ contains only a single element. 

Before we dive into the nonlinear setting, let us first consider a simplified LQ case to highlight the main challenges associated with the optimization of this loss.
In the LQ case, the \emph{evaluation} of the loss \eqref{eq:from_constrained_problem_to_unconstrained_one} is straightforward if there exists a closed-form expression for $p(\mathbf{y}_\mathcal{T}|\hat{\mathbf{x}})$, e.g., under a Gaussian observation model.
Even in that setting, however, it is important to realize that the problem remains nonconvex, as shown in Fig.~\ref{fig:contour_loss_2_player}.
The following proposition makes this challenge explicit, and the proof can be found in the Appendix.

\begin{figure}[t!]
    \centering
    \includegraphics[width=0.4\textwidth, trim = 0cm 8cm 0cm 6.5cm ]{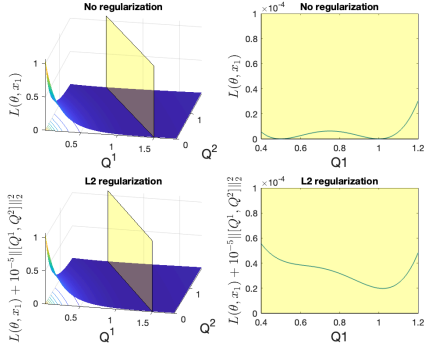}
    \caption{Visualization of the loss function $L(\theta,x_1)$ of the LQ game specified in \eqref{example:1D_loss_dynamics} and \eqref{example:1D_loss_costs}, and its $L_2$ regularization, with an initial condition $x_1=1$. We adopt Gaussian likelihood function. The yellow hyperplane is drawn according to $2Q^1+Q^2=3$. With $L_2$ regularization, the number of global minima is reduced.
    } 
    \label{fig:contour_loss_2_player}
\end{figure}

\begin{proposition}\label{prop:characterization of loss}
There exists an inverse LQ game problem \eqref{problem_formulation}: (a) whose global minima are isolated, and (b) for which there exist multiple cost functions that exactly match expert data from any initial condition, when there is no observation noise.
\end{proposition}
\begin{remark}
Proposition~\ref{prop:characterization of loss} does not imply that any inverse LQ game problem will suffer from the multiple global minima issue. Instead, Proposition~\ref{prop:characterization of loss} suggests that simply normalizing the cost vector does not rule out the possibility of having multiple global solutions. That is, there exist two cost parameter vectors which are linearly independent, but generate the same FBNE state trajectories for any given initial state. This non-injective mapping from the cost parameter space to the FBNE state trajectory space is a fundamental problem in inverse feedback games, and is not particular to the formulation \eqref{problem_formulation}. In practice, this multiple global minima issue could be mitigated by adding $L_2$ regularization, as visualized in Fig.~\ref{fig:contour_loss_2_player}. 
\end{remark}

Though being nonconvex, the loss function $\hat{L}(\theta,x_1)$ is differentiable with respect to both $\theta$ and $x_1$ under the condition of Theorem 3.2 in \cite{laine2021computation}, which follows from the implicit function theorem \cite{krantz2002implicit}. Inspired by the success of gradient-based methods in non-convex optimization with differentiable objective functions \cite{nesterov1983method,boyd2004convex,sutskever2013importance}, one natural idea is to apply gradient descent to minimize $\hat{L}(\theta,x_1)$. In what follows, we discuss efficient ways to evaluate and differentiate $\hat{L}(\theta,x_1)$ in nonlinear games. %

\subsection{Efficient Computation for a FBNE State Trajectory in Nonlinear Games}
It is easy to evaluate $\hat{L}(\theta,x_1)$ for LQ games, but when dynamics are nonlinear or objectives are non-quadratic, the problem becomes more challenging \cite{laine2021computation}. In forward games, this challenge can be addressed by using the ILQGames algorithm \cite{fridovich2020efficient}, which finds approximate local FBNE solutions in smooth non-LQ dynamic games. Given the effectiveness of this approximation scheme in those domains, we also adopt it as a submodule for evaluating the loss $\hat{L}(\theta,x_1)$. Akin to the ILQR method \cite{mayne1966second,LiTodorov2004iterative}, in each step of the ILQGames algorithm, the system dynamics $x_{t+1}=f(x_t,u_t)$ and the costs $\{g_t^i(x,u)\}_{t=1,i=1}^{T,N}$ are linearized and quadraticized, respectively, around a state trajectory $\mathbf{x}$ and a control trajectory $\mathbf{u}$. A FBNE strategy for each player of the derived LQ game is then used to update the state and control trajectories. This iteration continues until a convergence criterion is satisfied. 

To be more specific, we approximate $\hat{L}(\theta,x_1)$ by a new loss function $\tilde{L}(\theta,x_1)$ defined as,
\begin{equation}
\hat{L}(\theta, x_1)\simeq\tilde{L}(\theta,x_1):=-p\big(\mathbf{y}_\mathcal{T}|\mathbf{x}(\tilde{\mathbf{f}}_\theta, \tilde{\mathbf{g}}_{\theta}, x_1)\big)
\end{equation}
where $\mathbf{\state}(\tilde{\mathbf{f}}_\theta, \tilde{\mathbf{g}}_\theta, x_1)$ represents a FBNE state trajectory from initial condition $x_1$, for the LQ game defined by the linearized dynamics $\tilde{\mathbf{f}}_\theta$, quadraticized cost set $\tilde{\mathbf{g}}_\theta:=\{\tilde{\mathbf{g}}_{t,\theta}^i\}_{t=1,i=1}^{T,N}$ at the converged solution returned by ILQGames solver. Note that the linearized dynamics $\tilde{\mathbf{f}}_\theta$ depend upon $\theta$ via the state trajectory about which $\mathbf{f}$ is linearized; this trajectory is simulated under the feedback policy returned by ILQGames, where the policy depends upon costs $\mathbf{g}_\theta$.

\subsection{Differentiating the Loss in the Inverse Feedback Game Problem}\label{sec:gradient approximation}

The challenge of computing a feedback Nash equilibrium strategy not only makes the evaluation of the loss function $\hat{L}(\theta,x_1)$ hard, but also renders differentiation difficult. In this work, we approximate the gradient of $\hat{L}(\theta,x_1)$ using a similar idea as the ILQGames algorithm in the previous section. In other words, we propose to use the LQ approximation of the nonlinear game specified by $\tilde{\mathbf{f}}_\theta$ and $\tilde{\mathbf{g}}_\theta$ to derive an approximation to the gradient of $\hat{L}(\theta,x_1)$.  Note that $\tilde{g}_{t,\theta}^i(x,u) = \sum_{j=1}^{d_i} \theta_j^i \tilde{b}_{t,j,\theta}^i(x,u)$, where $\tilde{b}_{t,j,\theta}^i(x,u):\mathbb{R}^n\times \mathbb{R}^m\to \mathbb{R}$ is the $j$-th quadraticized cost basis function. 
By the chain rule, we have
\begin{equation*}
\begin{aligned}
    \frac{\partial\tilde{L}(\theta,x_1)}{\partial \theta_j^i} &= -\nabla_{\mathbf{x}} p(\mathbf{y}_\mathcal{T}|\mathbf{x})\Big|_{\mathbf{x}(\tilde{\mathbf{f}}_\theta, \tilde{\mathbf{g}}_\theta,x_1)}\cdot \frac{\partial \mathbf{x}(\tilde{\mathbf{f}}_\theta, \tilde{\mathbf{g}}_\theta, x_1)}{\partial \theta_j^i},\\
    \frac{\partial \mathbf{x}(\tilde{\mathbf{f}}_\theta, \tilde{\mathbf{g}}_\theta, x_1)}{\partial \theta_j^i}&=\Big(\nabla_{\tilde{\mathbf{f}}_\theta}\mathbf{x}(\tilde{\mathbf{f}}_\theta, \tilde{\mathbf{g}}_\theta, x_1)\frac{\partial \tilde{\mathbf{f}}_\theta }{\partial \theta_j^i} + \nabla_{\tilde{\mathbf{g}}_\theta}\mathbf{x}(\tilde{\mathbf{f}}_\theta,\tilde{\mathbf{g}}_{\theta}, x_1)\frac{\partial \tilde{\mathbf{g}}_\theta}{\partial \theta_j^i}    \Big).
\end{aligned}
\end{equation*}
The complexity of differentiating $\tilde{L}(\theta,x_1)$ comes from the fact that the linearized dynamics and the quadraticized costs are functions of $\theta$ implicitly, which makes the total derivative hard to compute. We propose to approximate the above gradient by treating the linearized $\tilde{\mathbf{f}}_\theta$ and each quadraticized cost basis function $\tilde{b}_{t,j,\theta}^i$ as constants with respect to $\theta$, denoted by $\tilde{\mathbf{f}}$ and $\tilde{b}_{t,j}^i$, 
and only compute the partial derivative with respect to $\theta$, rather than the total derivative:
\begin{equation*}\label{eq:GD_approximation}
    \frac{\partial \tilde{L}(\theta,x_1)}{\partial \theta_j^i}\simeq - \nabla_{\mathbf{x}} p(\mathbf{y}_{\mathcal{T}}| \mathbf{x})\Big|_{\mathbf{x}(\tilde{\mathbf{f}},\tilde{\mathbf{g}}_\theta,x_1)}\cdot \frac{\partial \mathbf{x}(\tilde{\mathbf{f}},  \{\sum_{j=1}^{d_i} \theta_j^i \tilde{b}_{t,j}^i\}_{t=1,i=1}^{T,N},x_1  )  }{\partial \theta_j^i}.
\end{equation*}
This is based on the observation that at the convergence of the forward ILQGames solver, the linearized dynamics are a good approximation of the full nonlinear dynamics $\mathbf{f}$, so long as the cost parameter being perturbed remains sufficiently small. We adopt a similar approximation for the gradient $\nabla_{x_1}\tilde{L}(\theta, x_1)$ by fixing the linearized dynamics and quadraticized costs and obtaining the partial derivative with respect to $x_1$. 

In summary, we approximate $\nabla \hat{L}(\theta,x_1)$ by $\nabla \tilde{L}(\theta, x_1)$. In practice, $\nabla \tilde{L}(\theta,x_1)$ can be efficiently computed by automatic differentiation \cite[Ch. 8]{nocedal2006numerical}. As exemplified in Fig.~\ref{fig:gradient descent quality}, the proposed gradient approximation is virtually always a descent direction and therefore aligns well with the true gradient of $ \hat{L}(\theta,x_1)$. %

~\\

\subsection{An Inverse Feedback Game Solver}

In this subsection, we present a solver for the inverse feedback game problem \eqref{problem_formulation}. In what follows, we first discuss how the three challenges mentioned in Section~\ref{sec: Problem Statement} are handled in our solver. We then introduce the proposed solver in Algorithm~\ref{Alg: Inverse Game Solution via Gradient Descent}.  

The first two challenges on noisy partial observation and incomplete trajectory data are handled by maintaining an estimate of the full initial condition and a noise-free state-input trajectory.
As shown in Section~\ref{sec: Results}, this procedure of joint reconstruction and filtering enables our solver to reliably recover player costs even in scenarios of substantial partial observability.
The third difficulty of evaluating and differentiating the objective function in the inverse \emph{feedback} game problem is mitigated by the efficient approximation outlined in Section~\ref{sec:gradient approximation}. To jointly infer the initial condition, the cost and the state-input trajectory, we first adopt the coordinate gradient descent method, where gradient descent steps are first taken over the initial condition $\hat{\state}_1$, and then taken over the cost parameter. We update the estimate of the noise-free full state-input trajectory by computing a FBNE state trajectory from the inferred initial condition and the cost.

We summarize our proposed solver in Algorithm~\ref{Alg: Inverse Game Solution via Gradient Descent}.
At the $k$-th iteration, we first compute an approximate FBNE state trajectory $\tilde{x}^{(k)}$ and the associated LQ approximation via the ILQGames algorithm of \cite{fridovich2020efficient}. Using this LQ approximation, we estimate $\nabla_{x_1}\hat{L}(\theta,x_1^{(k)})$ using the procedure outlined in Section \ref{sec:gradient approximation}. We then update the initial condition $x_1^{(k)}$ by a step of gradient descent, where the stepsize is chosen by a suitable linesearch technique \cite[Ch. 3]{nocedal2006numerical} such that the loss $\hat{L}(\theta,x_1)$ is sufficiently decreased. Given the updated initial condition $x_1^{(k+1)}$, we find a new approximate FBNE state trajectory via the ILQGames algorithm again, which is then used to estimate $\nabla_\theta \hat{L}(\theta^{(k)},x_1^{(k+1)})$ via the procedure in Section \ref{sec:gradient approximation}. With this gradient, we update $\theta^{(k)}$ by one step of gradient descent with linesearch. We repeat this procedure until, at convergence, we find a locally optimal solution $(\hat{\theta}, \hat{x}_1)$.

\begin{algorithm}[t!]
{
\small
\SetAlgoLined

\KwData{Horizon $\timehorizon > 0$, initial solution $\rewardparam^{(0)} \in \R^d$, observed time index set $\mathcal{T}\subseteq [T]$, observation data $\mathbf{y}_{\mathcal{T}}$, max iteration number $K$ , tolerance $\epsilon$.} 

\KwResult{Inferred cost parameter $\hat{\theta}$ and $\hat{x}_1$}

\For{$k = 0,1,\ldots, K$}{
$(\tilde{\mathbf{\state}}^{(k)},  \{\tilde{\gamma}_t^i\}_{t=1,i=1}^{T,N}, \tilde{\mathbf{f}}_{\theta^{(k)}}, \tilde{\mathbf{g}}_{\theta^{(k)}} )\gets \textrm{ILQGames}(\mathbf{f}, \mathbf{g}_{\theta^{(k)}}, x_1^{(k)})$ \label{alg:iLQ to x0}

$\nabla_{x_1} \hat{L}(\theta^{(k)},x_1^{(k)})\gets $ evaluated using $\tilde{\mathbf{f}}_{\theta^{(k)}}$ and $\tilde{\mathbf{g}}_{\theta^{(k)}}$ via Gradient Approximation in Section~\ref{sec:gradient approximation}\label{alg:evaluation of GD x0}

$x^{(k+1)}_1\gets x^{(k)}_1 - \eta \nabla_{x_1} \hat{L}(\theta^{(k)},x_1^{(k)}) $ with line search over $\eta$\label{alg:GD x0}

$(\check{x}^{(k)}, \{\check{\gamma}_t^i\}_{t=1,i=1}^{T,N},\check{\mathbf{f}}_{\theta^{(k)}}, \check{\mathbf{g}}_{\theta^{(k)}} )\gets \textrm{ILQGames}\big(\mathbf{f}, \mathbf{g}_{\theta^{(k)}}, x_1^{(k+1)}\big)$\label{alg:update traj}

$\nabla_\theta \hat{L}(\theta^{(k)}, x_1^{(k+1)})\gets$ evaluated using $\check{\mathbf{f}}_{\theta^{(k)}}$ and $\check{\mathbf{g}}_{\theta^{(k)}}$ via Gradient Approximation in Section~\ref{sec:gradient approximation}\label{alg:evaluation of GD theta}

$ \theta^{(k+1)}\gets \theta^{(k)} -\eta' \nabla_\theta\hat{L}(\theta^{(k)}, x_1^{(k+1)})$ with line search over $\eta'$\label{alg:GD theta}

\textbf{Return} $(\theta^{(k+1)},x_1^{(k+1)})$ if $\|\theta^{(k)}-\theta^{(k-1)}\|_2\le \epsilon$ or \textbf{Return} $(\theta^{(k')},x_1^{(k')})$, where $k'\gets \arg\min_k \tilde{L}(\theta^{(k)},x_1^{(k)})$, if iteration number $k$ reaches $K$.

}

\caption{Inverse Iterative LQ (i$^2$LQ) Games}
 \label{Alg: Inverse Game Solution via Gradient Descent}
}
\end{algorithm}

\section{Experiments}
\label{sec: Results}
In this section, we adopt the open-loop solution method of \cite{peters2021inferring} as the baseline method and compare it to Algorithm 1. In particular, we evaluate Algorithm \ref{Alg: Inverse Game Solution via Gradient Descent} in several Monte Carlo studies which aim to justify the following claims.

\begin{itemize}
    \item The proposed gradient approximation often aligns with a descent direction in the loss function. 
    
    \item Algorithm 1 is more robust than the open-loop baseline method \cite{peters2021inferring} with respect to noise in, and incomplete observations of, the expert demonstration trajectory.
    
    \item The cost functions inferred by Algorithm~\ref{Alg: Inverse Game Solution via Gradient Descent} can be generalized to predict trajectories from unseen initial conditions.
    \item Algorithm~1 can infer nonconvex costs in nonlinear games.
\end{itemize}






\subsection{Gradient Approximation Quality}
We continue the 2-vehicle platooning example defined in \eqref{eq:dubins_car} and \eqref{eq:2_cars_costs}. We measure the performance of Algorithm~\ref{Alg: Inverse Game Solution via Gradient Descent} in two settings, incomplete expert trajectory data with noisy partial state observation, and complete expert trajectory data with noisy full observation. In the first case, each player's partial observation only contains its x-position, y-position and heading angle. The time index set of the incomplete trajectory is $\mathcal{T}=[T]\setminus\{11,12,\dots,19\}$. In the second case, the expert data includes the noisy observation of all the states of both players at all $t\in[T]$. The ground truth expert state trajectory follows a FBNE strategy from the initial condition $x_1=[0,0.5,\frac{\pi}{2}, 1,1,0,\frac{\pi}{2},1]$ and the target lane is $p_x^*=0.0$. At each variance level $\sigma\in\{0.004,0.008,\dots,0.04\}$, we generate 10 noisy observations of the ground truth expert trajectory, with isotropic zero-mean Gaussian noise. For each noisy expert data set $\mathbf{y}_\mathcal{T}$, we minimize the negative log-likelihood objective in \eqref{problem_formulation}, i.e., $\sum_{t\in\mathcal{T}} \|y_t - h(x_t) \|_2^2$, where $h(\cdot):\mathbb{R}^n\to \mathbb{R}^\ell $ maps a state $x_t$ to its partial observation.

As shown in Fig.~\ref{fig:gradient descent quality}, the loss decreases monotonically on the average. This indicates that the gradient approximation proposed in Section \ref{sec:gradient approximation} provides a reliable descent direction. The inverse feedback game problem becomes challenging when there is only partial state observation and incomplete trajectory data, and the quality of inferred costs may degrade when the observation noise is high. 


\subsection{Robustness, Generalization and the Ability to Infer Nonconvex Costs}
We continue the previous 2-vehicle example and compare Algorithm~\ref{Alg: Inverse Game Solution via Gradient Descent} and the baseline in a Monte Carlo study, where we infer the costs under 10 different levels of Gaussian noise with increasing variance. In particular, we evaluate three metrics in Fig.~\ref{fig:cars2_comparison}: (a) the distance between the noisy expert data and the FBNE state trajectory which results from players' inferred costs; (b) the distance between the computed FBNE state trajectory (under the players' inferred costs) and the ground truth expert data. 
An example of such a comparison is shown in Fig.~\ref{fig:2_cars_prediction}. Finally, we evaluate (c) the distance between the inferred FBNE state trajectories and the FBNE state trajectory under the ground truth costs for some randomly sampled initial conditions, which is also visualized in Fig.~\ref{fig:2_cars_generalization}. Collectively, the results demonstrate that \emph{Algorithm~\ref{Alg: Inverse Game Solution via Gradient Descent} has better robustness and generalization performance than the open-loop baseline when the expert data follows the FBNE assumption.}

To show that Algorithm 1 can infer nonconvex cost functions, we extend the previous 2-vehicle platooning example and assume that the 2-vehicle team encounters a third vehicle and the follower wants to stay close to the leader without colliding with the third vehicle. We model this scenario as a 3-vehicle game with a 12 dimensional state space and a horizon $T=30$. The dynamics for each vehicle is the same as \eqref{eq:dubins_car} and the costs are as follows,
\begin{equation*}
    \begin{aligned}
    \costcomponent_{t,\theta}^1(x_t,u_t) =&\theta_1^1\|p_{x,t}^1\|_2^2+\theta_2^1\|p_{x,t}^2-p_x^*\|_2^2+\|v_t^1-2\|_2^2 \\&+\|\beta_t^1-\frac{\pi}{2}\|_2^2+ \|u_t^1\|_2^2\\
    \costcomponent_{t,\theta}^2(x_t,u_t) =& \theta_1^2\|p_{x,t}^2\|_2^2+\|\beta_t^2-\frac{\pi}{2}\|_2^2+\theta_2^2\| p_{x,t}^2-p_{x,t}^1 \|_2^2+\|v_t^2-2\|_2^2\\&-\frac{1}{2}\log(\|p_{x,t}^2-p_{x,t}^3\|_2^2+\|p_{y,t}^2-p_{y,t}^3\|_2^2) +\|u_t^2\|_2^2\\
    \costcomponent_{t,\theta}^3(x_t,u_t)=&\theta_1^3\|p_{x,t}^3-\frac{1}{2}\|_2^2 + \|u_t^3\|_2^2
    \end{aligned}\label{examples:3cars}
\end{equation*}
where the ground truth $\theta^*\in\mathbb{R}^5$ is $[0,4,0,4,2]$. The ground truth expert state trajectory follows a FBNE strategy from the initial condition $x_1=[0,1,\frac{\pi}{2},2,0.3,0,\frac{\pi}{2}, 2,0.5,0.5,\frac{\pi}{2},2 ]$, where the last four elements encode the state of the third vehicle. The target lane in the expert data is $p_x^*=0.2$. 

Similar to the 2-vehicle experiment, we consider two settings, incomplete trajectory data with partial state observation and complete trajectory data with full state observation. The partial state observation includes all the states of each vehicle except for the velocity of all the vehicles, and the time indices set of the incomplete trajectory is $\mathcal{T}=[T]\setminus\{11,12,\dots,19\}$. 
The nonconvex cost of player 2 causes numerical problems in the baseline KKT OLNE solver \cite{peters2021inferring}. Thus, we add an $L_2$ regularization $10^{-4}\|\theta\|_2^2$ to the loss $\hat{L}(\theta,x_1)$ and summarize the Monte Carlo study in Fig.~\ref{fig:cars3_comparison}, where we see Algorithm~\ref{Alg: Inverse Game Solution via Gradient Descent} is also able to learn better cost functions reflecting the true intentions of each vehicle in feedback games, even with only partial state observations and incomplete trajectory data. 


\begin{figure}[t!]
    \centering
    \includegraphics[width=0.44\textwidth, trim=0cm 2cm 0cm 0cm]{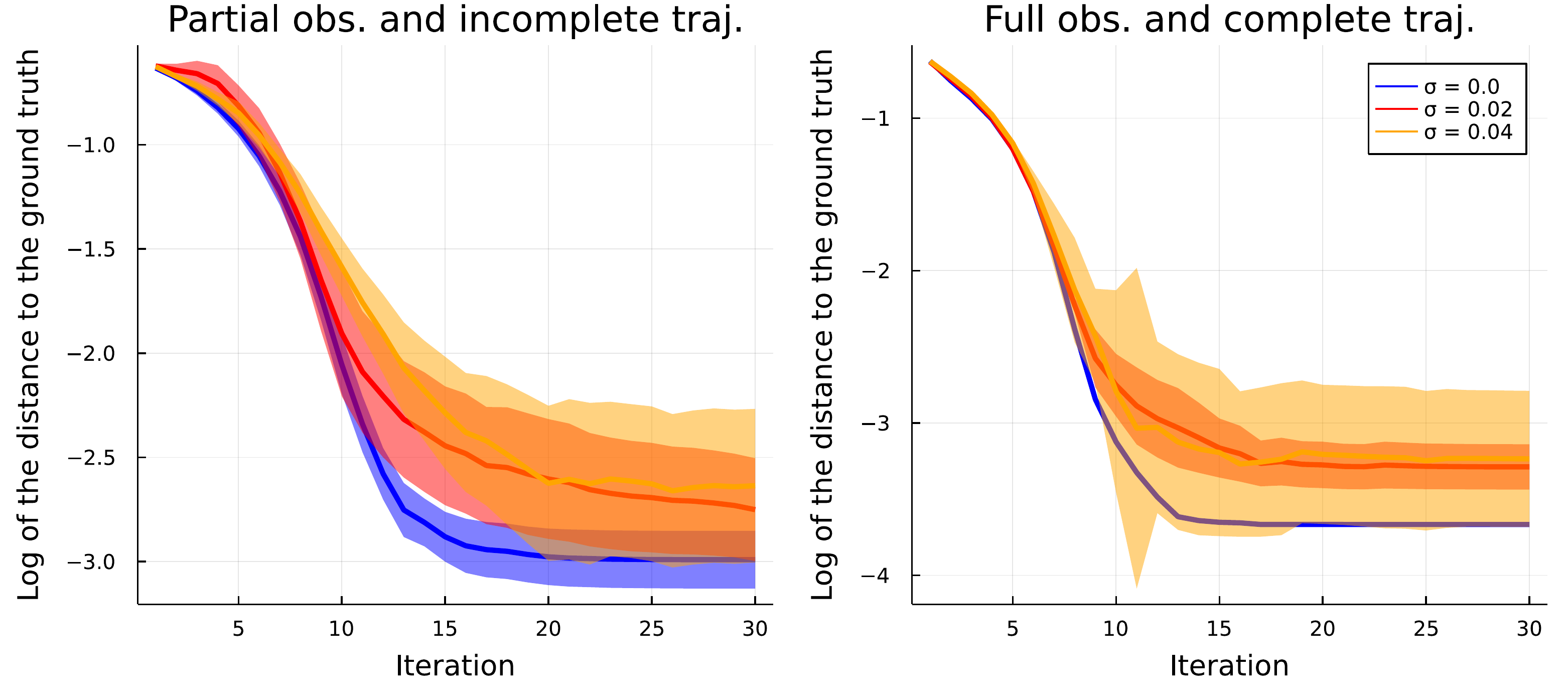}
    \caption{Convergence of Algorithm 1 with the Gradient Approximation proposed in Section \ref{sec:gradient approximation}. The loss decreases monotonically on the average. The bold lines and shaded areas represent the mean values and their standard error, i.e., the variance divided by the square root of the sample size, respectively.}
    \label{fig:gradient descent quality}
\end{figure}

\begin{figure}[t!]
    \centering
    \includegraphics[width=0.44\textwidth, trim=0cm 2cm 0cm 1cm]{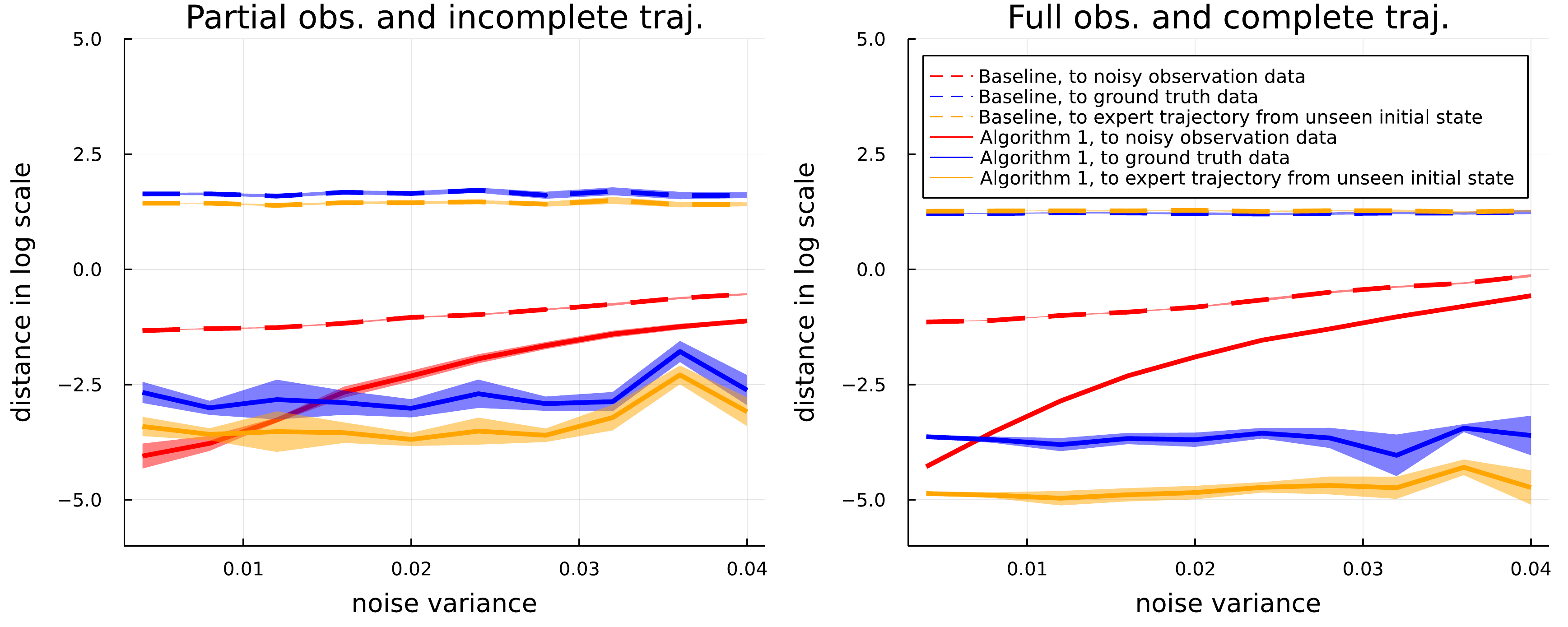}
    \caption{2-vehicle platooning scenario. The bold lines and shaded areas represent the mean values and their standard error, i.e., the variance divided by the square root of the sample size, respectively. As the noise variance growing, the converged loss value increases, as shown in the red curves. However, Algorithm~\ref{Alg: Inverse Game Solution via Gradient Descent} is still able to learn a more accurate cost and has less generalization error than the baseline, as shown in the blue and yellow curves, respectively.}
    \label{fig:cars2_comparison}
\end{figure}

\begin{figure}[t!]
    \centering
    \includegraphics[width=0.44\textwidth, trim = 0cm 1.5cm 0cm 0cm]{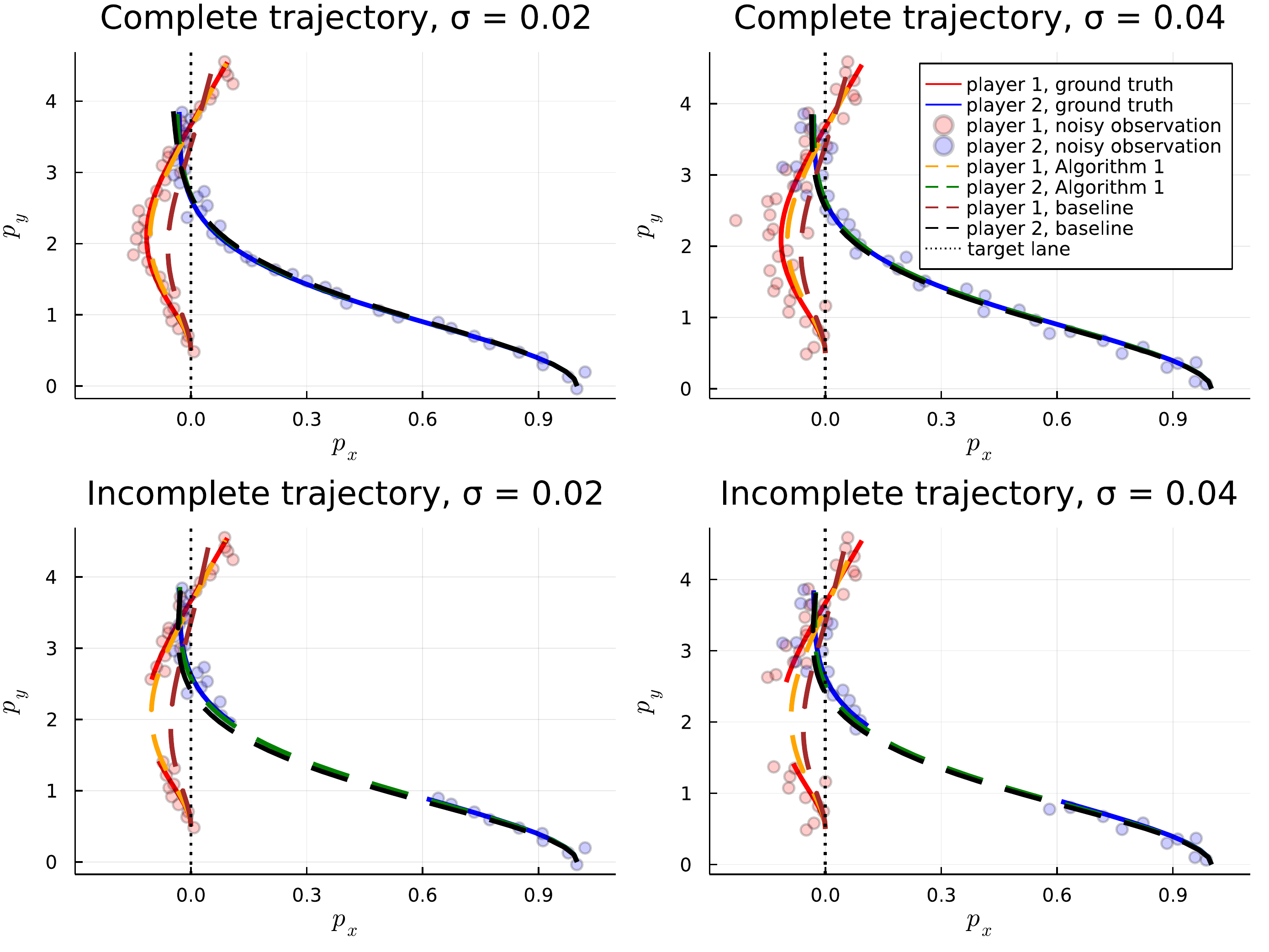}
    \caption{Full and partial, noisy observation of the expert trajectories. Dashed lines represent predicted trajectories which result from inferred costs, and solid lines are ground truth. The trajectories predicted by Algorithm~\ref{Alg: Inverse Game Solution via Gradient Descent} are closer to the ground truth than the baseline.}
    \label{fig:2_cars_prediction}
\end{figure}
\begin{figure}[t!]
    \centering
    \includegraphics[width=0.44\textwidth, trim=0cm 1.5cm 0cm 0cm]{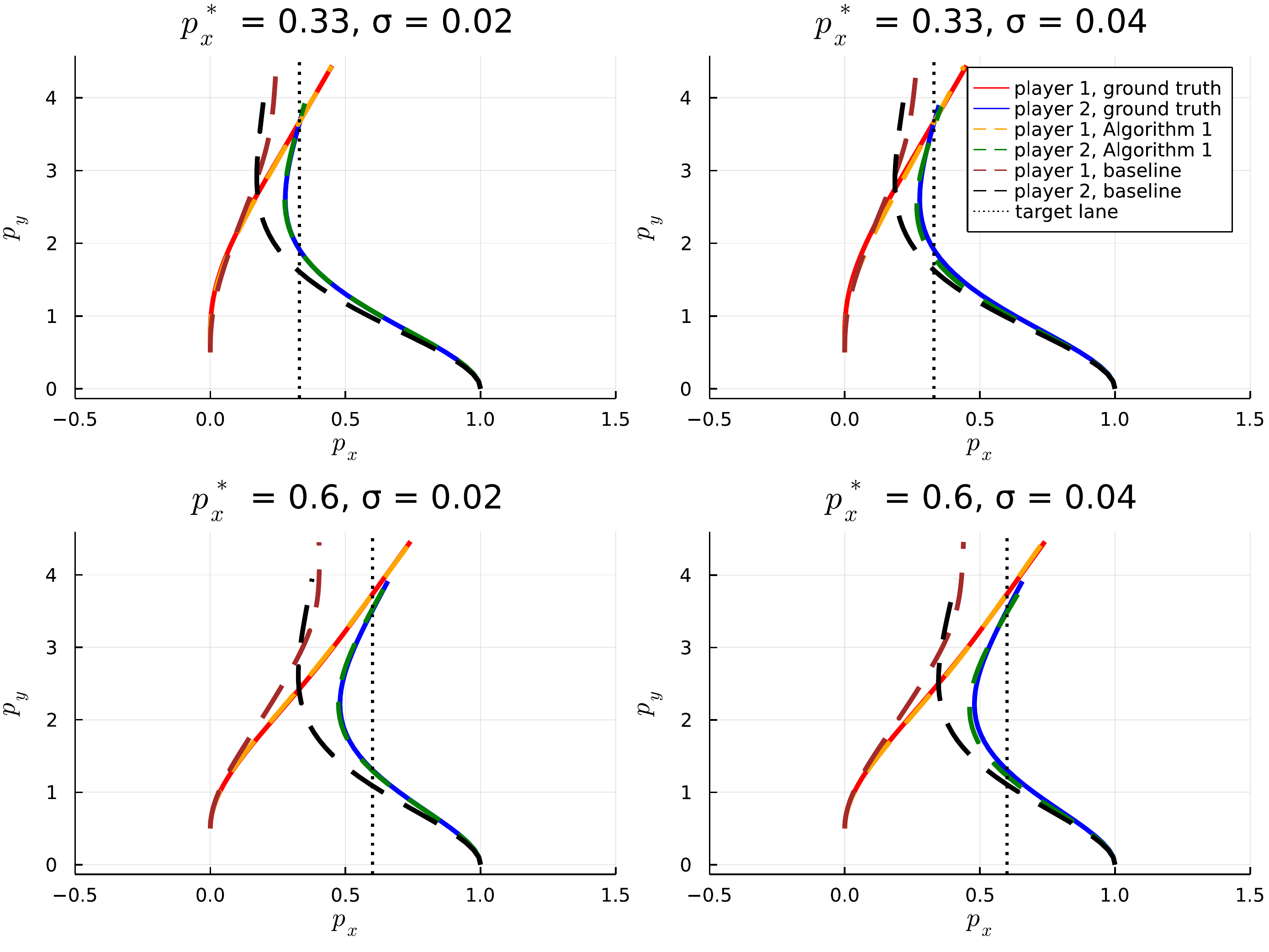}
    \caption{Generalization performance comparison. $p_x^*$ is the target lane position that player 1 wants to guide player 2 toward. All the costs are inferred from partial observations and incomplete trajectory data, with different noise variance specified in each of the subplot. The trajectories predicted by Algorithm~\ref{Alg: Inverse Game Solution via Gradient Descent} are closer to the ground truth than the baseline.}
    \label{fig:2_cars_generalization}
\end{figure}

\begin{figure}[t!]
    \centering
    \includegraphics[width=0.44\textwidth, trim= 0cm 2cm 0cm 1cm]{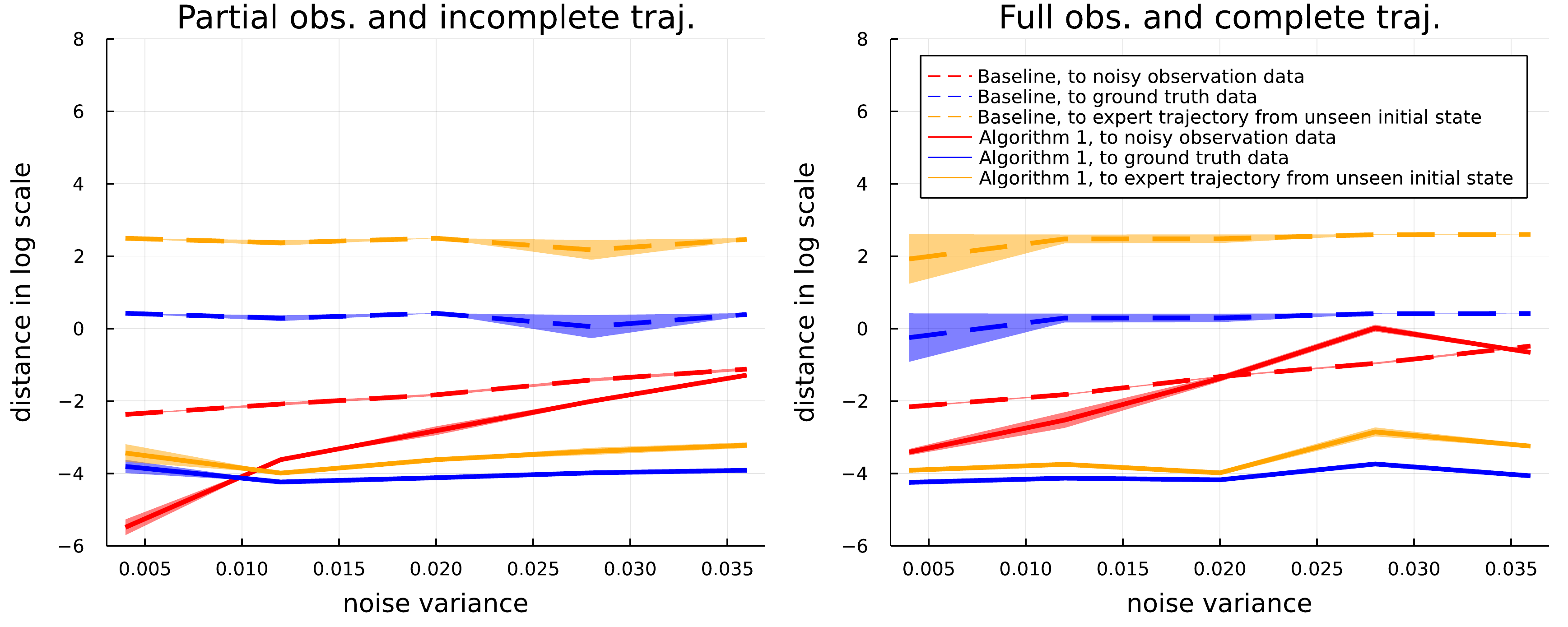}
    \caption{3-vehicle platooning scenario. The bold lines and shaded areas represent the mean values and their standard error, i.e., the variance divided by the square root of the sample size, respectively. As the noise variance growing, the converged loss value increases on the average, as shown in the red curves. However, Algorithm~\ref{Alg: Inverse Game Solution via Gradient Descent} is still able to learn a more accurate cost and has less generalization error than the baseline, as shown in the blue and yellow curves, respectively.}
    \label{fig:cars3_comparison}
\end{figure}

\label{subsec: Discussion}

\section{Conclusion}
\label{sec: Conclusion and Future Work}
In this work, we propose an efficient cost inference algorithm for inverse feedback nonlinear games, with only partial state observation and incomplete trajectory data. 
Empirical results show that the proposed solver converges reliably for inverse games with nonconvex costs and has superior generalization performance than a state-of-the-art open-loop baseline method when the expert demonstration reflects a group of agents acting in a dynamic feedback game. There are many future directions. We can investigate under what conditions the cost can be inferred exactly in feedback games. The active and online inference are also promising directions. In addition, we are eager to extend this work to settings of closed-loop interaction. In such an extension, rather than merely inferring the objectives of observed players, this information would be used to guide the decision-making of an autonomous agent in that scene.

\section*{Appendix}\label{sec: Appendix}

\begin{proof}[Proof of Proposition~\ref{prop:characterization of loss}]
Proposition 1 claims that there exists an inverse LQ game, which has isolated global minima and the induced FBNE state trajectories of those solutions match the expert demonstration. Here, we show such a counterexample, which supports the claim. Consider a 2-player horizon-3 LQ game with the linear dynamics
\begin{equation}\label{example:1D_loss_dynamics}
    x_{t+1} = x_t+u_t^1 + u_t^2, \ \ t\in\{1,2,3\},
\end{equation}
and the cost
\begin{equation}\label{example:1D_loss_costs}
\begin{aligned}
    \costcomponent_{t}^1(\state_t,\ctrl_t) &= \frac{1}{2}(Q^1 \|\state_t\|_2^2 +\|\ctrl_t^1\|_2^2),\ \ t\in \{1,2\},\\
    \costcomponent_{t}^2(\state_t,\ctrl_t) &=\frac{1}{2} (Q^2 \|\state_t\|_2^2 +2\|\ctrl_t^2\|_2^2),\ \ t\in\{1,2\},\\
    \costcomponent_3^1(\state_3,\ctrl_3) &= \frac{1}{2}Q^1\|\state_3\|_2^2,\ \costcomponent_3^2(\state_3,\ctrl_3) = \frac{1}{2}Q^2\|\state_3\|_2^2.
\end{aligned}
\end{equation}
We assume that the ground truth solutions are $Q^1=1$, $Q^2=1$. We will show there is also one extra solution $\hat{Q}^1=\frac{1}{2}$ and $\hat{Q}^2=2$, which yields the same FBNE state trajectory as the ground truth for any initial condition. We follow the same definition of the variable $\{Z_t^i\}_{t=1,i=1}^{3,2}$ as in \cite{basar1998DynamicNoncooperativeGameTheory}. 
By definition, we have $Z_t^i\ge Q^i>0$, when $Q^1\in\mathbb{R}_+$ and $Q^2\in\mathbb{R}_+$. Following the notations in FBNE condition in Corollary 6.1 of \cite{basar1998DynamicNoncooperativeGameTheory}, we consider the feedback matrices $\{P_t^i\}_{t=1,i=1}^{2,2}$, 
\begin{equation}
    \begin{bmatrix}
    P_t^1\\ P_t^2
    \end{bmatrix}=\underbrace{\begin{bmatrix}
    1 + Z_{t+1}^1 & Z_{t+1}^1\\ Z_{t+1}^2 & 2+Z_{t+1}^2
    \end{bmatrix}}_{G_t^i} \begin{bmatrix}
    Z_{t+1}^1 \\ Z_{t+1}^2
    \end{bmatrix},\ \  \forall t\in \{1,2\},
\end{equation}
where the matrix $G_t^i$ is invertible because $\det (G_t^i) = 2+Z_{t+1}^2 +2Z_{t+1}^1 >0$. The above analysis suggests that the FBNE state trajectory for all $Q^1>0$ and $Q^2>0$ are uniquely determined. We consider the time instant $t=2$, and observe
\begin{equation}
    \begin{bmatrix}
    P_{2}^1 \\ P_{2}^2
    \end{bmatrix} = \begin{bmatrix}
    1+Q^1 & Q^1 \\ Q^2 & 2+Q^2
    \end{bmatrix}^{-1} \begin{bmatrix}
    Q^1\\ Q^2
    \end{bmatrix}=\frac{1}{2+2Q^1 + Q^2}\begin{bmatrix}2Q^1 \\ Q^2
    \end{bmatrix}.
\end{equation}
We then have the closed-loop dynamics $x_{3} = (1-P_{2}^1-P_{2}^2)x_{2}
 =\frac{2}{2+2Q^1 + Q^2}x_{2}$,
which yields that for two pairs of positive variables $(Q^1,Q^2)$ and $(\hat{Q}^1, \hat{Q}^2)$, a necessary condition for them to have the same FBNE trajectory is that $2Q^1+Q^2 = 2\hat{Q}^1 + \hat{Q}^2$. We have $Z_{2}^1 = Q^1 + \frac{Q^1+(2Q^1)^2}{(2+2Q^1+Q^2)^2},\ 
    Z_{2}^2 = Q^2 + \frac{Q^2+2(Q^2)^2}{(2+2Q^1+Q^2)^2}$.
Similarly, for the time instant $t=1$, we have $x_{2}=(1-P_{1}^1-P_{1}^2)x_{1} = \frac{2}{2+2Z_2^1 + Z_2^2}x_{1}$.
A necessary condition for $(\hat{Q}^1,\hat{Q}^2)$ to have the same FBNE state trajectory as $(Q^1,Q^2)$ is that the following 2 equations are satisfied,
\begin{equation}\label{eq:2_LQ_equation}
    \begin{aligned}
        &2Q^1+Q^2 = 2\hat{Q}^1 + \hat{Q^2}\\
        &2\big(Q^1+\frac{Q^1+(2Q^1)^2}{(2+2Q^1+Q^2)^2} \big) + Q^2+ \frac{Q^2+2(Q^2)^2}{(2+2Q^1+Q^2)^2}\\ 
        &\ \ \ \ = 2\big(\hat{Q}^1+\frac{\hat{Q}^1+(2\hat{Q}^1)^2}{(2+2\hat{Q}^1+\hat{Q}^2)^2} \big) + \hat{Q}^2+ \frac{\hat{Q}^2+2(\hat{Q}^2)^2}{(2+2\hat{Q}^1+\hat{Q}^2)^2}.
    \end{aligned}
\end{equation}
We substitute $Q^1=1$, $Q^2=1$ and $\hat{Q}^2=3-2\hat{Q}^1$ into the second row of \eqref{eq:2_LQ_equation}, which is reduced to a 2-degree polynomial of $\hat{Q}^2$. By the fundamental theorem of algebra \cite{cauchy1821cours}, there exist at most 2 solutions for $\hat{Q}^2$. The two pairs of $(\hat{Q}^1,\hat{Q}^2)$ satisfying \eqref{eq:2_LQ_equation} are $(1,1)$ and $(\frac{1}{2},2)$. The two global minima are isolated. Since the dimension of the state $x_t$ is 1, for all initial states $x_1\in\mathbb{R}$, the FBNE state trajectories under the costs specified by the two pairs cost parameters $(1,1)$ and $(\frac{1}{2}, 2)$ coincide with each other.
\end{proof}

\bibliographystyle{plain}

\bibliography{references}

\end{document}